\documentclass[preprint,authoryear,12pt]{elsarticle}

\usepackage[authoryear]{natbib}
\usepackage{amssymb,latexsym}
\usepackage{graphicx,amsmath,amsfonts,amsthm}
\usepackage{comment}

\def\row#1#2{{#1}_1,\ldots ,{#1}_{#2}}

\input prepictex
\input postpictex
\input pictexwd

\newcounter{theexample} 

\newcounter{themyclaim}
\newtheorem{myclaim}[themyclaim]{Claim}

\newcommand\bibsection{%

\section*{\bibname\markright{\MakeUppercase{\bibname}}}}

\journal{ArXiv.org}

\def\row#1#2{{#1}_1,\ldots ,{#1}_{#2}}

\def\2vec#1#2{\left(\begin{array}{c}{#1}\\{#2}\end{array}\right)}

\newtheorem{theorem}{Theorem}
\newtheorem{corollary}{Corollary}
\newtheorem{lemma}{Lemma}

\newtheorem{proposition}{Proposition}
\newtheorem{example}{Example}

\newtheorem{definition}{Definition}

\usepackage{multirow}
\usepackage{colortbl}
\usepackage[usenames,dvipsnames]{xcolor}
 
\begin{document}

\begin{frontmatter}
\title{\bf Is it ever safe to vote strategically?}
\author{\bf  Arkadii Slinko and Shaun White}
\address{The University of Auckland, New Zealand}


\begin{abstract}
There are many situations in which mis-coordinated strategic voting can leave strategic voters worse off than they would have been had they not tried to strategize.  We analyse the simplest of such scenarios, in which the set of strategic voters all have the same sincere preferences and all cast the same strategic vote, while all other voters vote sincerely.   Most mis-coordinations in this framework can be classified as instances of either strategic overshooting (too many voted strategically) or strategic undershooting (too few).  If mis-coordination can result in strategic voters ending up worse off than they would have been had they all just voted sincerely, we call the relevant strategic vote unsafe.  We show that under every onto and non-dictatorial social choice rule there exist circumstances where a voter has an incentive to cast a safe strategic vote.  We extend the Gibbard-Satterthwaite Theorem by proving that every onto and non-dictatorial social choice rule can be individually manipulated by a voter casting a safe strategic vote.
\end{abstract}
\begin{keyword}
Social choice function \sep strategic voting \sep Gibbard-Satterthwaite theorem\sep manipulation \sep safe strategic vote.
\end{keyword}

\end{frontmatter}




\newpage


\section{Introduction}

\subsection{Overview}
In this paper we consider strategic behaviour of voters when one candidate must be chosen from a set of candidates  and voters' preferences with respect to those candidates are expressed as strict linear orders. The Gibbard-Satterthwaite Theorem (GST)---arguably one of the most important theoretical results  in this direction to date---states that under any onto and non-dictatorial social choice rule there exists a situation where a voter can do better by casting a strategic vote rather than the sincere one, provided everyone else votes sincerely.  In short, every onto and non-dictatorial social choice rule is individually manipulable.  If only one voter can manipulate the election, and she is rational, then she will do this unless having some moral objections towards manipulation.  However, it may happen that two voters with differing preferences each can manipulate the same election. In this case the result is much less predictable. We illustrate this by the following example  

\begin{example}
\label{ex1}
Suppose four people are to choose between three alternatives.
Let the profile of sincere preferences be
\[
\begin{array}{c|c|c|c}
1&2&3&4\\
\hline
a&b&c&c\\
b&a&a&b\\
c&c&b&a\\
\hline
\end{array}
\]
and the rule used be Plurality with breaking ties in accord with the order $a>b>c$. If everybody votes sincerely, then $c$ is elected. Voters 1 and 2 are Gibbard-Satterthwaite manipulators. Voters 3 and 4 are expected to vote sincerely. Voter 1 can make $b$ to win by voting $b>a>c$ and voter 2 can make $a$ to win by voting $a>b>c$. However, if they both try to manipulate, $c$ will remain the winner. Each of them would prefer that the other one manipulates.
\end{example}
 In general, the situation when voters of different types interact strategically is too complex to analyse. Going from the Gibbard-Satterthwaite framework, when only one voter can be strategic to the framework where all voters may be strategic is too big a jump. In this paper we analyse an intermediate framework, namely the case when only voters of one particular type may be strategic. Let us see why this framework is interesting.

If voting is anonymous and one voter can individually manipulate, so can every other voter present with the same preferences.  And in a large profile there will be many such voters. But it is not in general true that if multiple like-minded voters simultaneously attempt exactly the same individual manipulation, the outcome will necessarily be (from their perspective) a success.  To manipulate successfully the strategically-inclined voter may need to coordinate (intentionally or accidentally) with other voters with whom she shares the same preferences and, hence, the same incentives.  Here is an example when such coordination is needed.

\begin{example}
\label{ex2}
Suppose four people are to choose between three alternatives.
Let the profile of sincere preferences be
\[
\begin{array}{c|c|c|c}
1&2&3&4\\
\hline
a&a&b&c\\
b&b&c&b\\
c&c&a&a\\
\hline
\end{array}
\]
and the rule used be Borda with breaking ties in accord with the order $a>b>c$. If everybody votes sincerely, then $b$ is elected. Voters 1 and 2 are Gibbard-Satterthwaite manipulators. Voter 1 can make $a$ to win by voting $a>c>b$ and voter 2 can do the same. However, if they both try to manipulate, their worst alternative $c$ will become the winner.  
\end{example}

We note that in the second example it is easier for the two would-be manipulators to coordinate since they have the same type, i.e., have identical sincere preferences, and hence no conflict of interest. The conclusion that we make from consideration of these two examples is: when a group of voters involves voters of different types then, for them, manipulation is not purely coordination problem, negotiation may be also required. This is why in this paper we will consider only homogeneous groups of voters, that is voters with identical preferences.   \par\medskip

As in the Gibbard-Satterthwaite farmework we assume that voters have complete information about voters' sincere preferences and beliefs about their voting intentions. To extend the Gibbard-Satterthwaite theorem we must consider the case when  several voters of the same type are all pivotal and can change the result to their advantage by casting the same strategic vote\footnote{This will be for example the case if the voting rule is anonymous.} believing that all voters of other types will vote sincerely. In such situation any of those would-be manipulators may find it unsafe to act on such incentive. 
 Indeed, it may happen that if too many or too few of them act on the same incentive the result may be worse than the outcome when everyone is sincere. This is exactly the situation in Example~\ref{ex2}. An alternative to this may be a situation when one or more of those would-be manipulators will not be worse of no matter how other voters of her type will vote; in such case we say that she can cast a safe strategic vote. Of course, the term `safe' refers only to situations when voters of all other types vote sincerely.  The question that arises is whether or not under any non dictatorial voting rule there exist situations when a voter can manipulate and do it  safely.  In this paper we will answer this question in the affirmative. 

The main result of this paper, just formulated, is about pivotal voters, however, most of the paper we will work with voters who are not pivotal. Our starting point will be the concept of an incentive to vote strategically.  We will say that a voter has {\em  an incentive to vote strategically} if she can become part of a group of voters of her type who all would benefit if every group member casts the same strategic vote.  The rational for this definition is that strategic voters despite not being pivotal vote strategically because they hope or expect that they will be joined by the `right' set of like-minded voters also casting strategic votes, and a more favourable outcome will be brought about by their collective effort. 
For example,  at the 2000 US Presidential election, a voter who preferred Nader to Gore and Gore to Bush, and felt assured that Nader would not win, would have an incentive to strategically vote Gore. This incentive would have been shared by every voter of this particular type but not felt by any voter of another type.

As in the case of pivotal voters, we can distinguish between {\em safe} and {\em unsafe} strategic votes.  In some situations, all like-minded voters in an election all share an incentive to vote strategically but mis-coordination can occur: if the group splits up into those who act on the incentive and those who do not  in the `wrong' way, the social choice outcome can, from their perspective, deteriorate rather than improve.  In such cases we say the strategic vote in question is unsafe.  The most recognisable unsafe strategic votes are those that can lead to strategic overshooting or strategic undershooting.  {\em Strategic overshooting} occurs when `too many' like-minded voters simultaneously act on a common incentive to vote strategically, and {\em strategic undershooting} occurs when `too few' do.  When voting is not anonymous, a strategic vote can be unsafe even though neither over- nor undershooting are possible.
In the example of the 2000 US Presidential election the strategic vote by the Nader voters for Gore  would have been safe.  Despite of this, the number of voters having this incentive and acting on it was insufficient to bring about Gore's presidency\footnote{In the crucial state Florida Bush got only 537 more votes than Gore while 97,488 voters voted Nader  \cite{US2000results}.}. 

It is worthwhile to note that in this framework a sincere vote can be strategic. Indeed a strategically inclined voter may decide that there is already sufficiently many voters who will cast a strategic vote so she may end up submitting her sincere vote after weighing all pros and cons. 

Consideration of the difference between a safe and an unsafe strategic vote reveals an additional reason why a voter might not act on a known to her incentive to vote strategically: without the ability to 
coordinate with their like-minded fellows, the strategic voter may face the risk of making matters worse rather than better.  This adds to the list of impediments to strategic voting identified elsewhere, which includes ($i$) the inability to acquire the necessary information~\citep{ChopraPacuitParikh2004}, ($ii$) the inability to process the large amount of information one has~\citep{BartoldiToveyTrick1989,ConitzerSandholmLang2007}, ($iii$) the possibility of provoking an unfavourable response~\citep{Pattanaik1976a,Pattanaik1976b,Barbera1980}, and ($iv$) lack of familiarity with the electoral context~\citep{CrispOlivellaPotter2012}.

The first main result of this paper will be Theorem~\ref{ASSWtheorem} that states that under any onto and non-dictatorial social choice rule there can arise a scenario where a voter has an incentive to cast a safe strategic vote. 
Our second main result---which extends the Gibbard-Satterthwaite theorem---states that under any onto and non-dictatorial social choice rule there will exist a situation where a certain pivotal voter can manipulate safely. This can be deduced from Theorem~\ref{ASSWtheorem} rather straightforwardly.

This paper is organised as follows.  In Section~\ref{sec:rel_lit} we discuss the literature related to this article. Section~\ref{sec:model} gives our formal definitions and states our results.  Our Section~\ref{sec:examples} will be entirely devoted to examples of safe and unsafe strategic votes and their graphical illustration. 
The main theorems are proven in Sections~\ref{sec:tools} -~\ref{sec:fouralternatives}.  
Section~\ref{sec:conclusion} concludes with the discission of the results obtained and formulates some open questions.

This paper is a refinement of our preprints~\citep{SlinkoWhite2008a,SlinkoWhite2008b}.


\subsection{Related literature}
\label{sec:rel_lit}

There are three strands in the literature  related to this work: coalitional manipulability of single-winner social choice rules, manipulability of multi-winner social choice rules  and informational aspects of manipulability. 


An election is called coalitionally manipulable \cite{Murakami1968,Pattanaik1973,Chamberlin1985} if a group of voters may change their votes in such a way that the result of the election becomes more desirable for them.  It is also assumed that the voters outside the coalition are not strategic and vote sincerely. 

The central concept for this strand of literature is the concept of the average size of minimal manipulating coalition \cite{Chamberlin1985}. This parameter is used to compare the `quality' of voting rules \cite{Chamberlin1985,PS2006,Pritchard2009}. 
It has been shown that, if voters are assumed to be independent and the size of the coalition grows slower than $\sqrt{n}$, where $n$ is the number of voters, then the probability that a random election is manipulable by the coalition of this size goes to zero as the number of voters goes to infinity, whereas if the number of manipulators grows faster than $\sqrt{n}$, then the probability that a random profile is manipulable goes to one (see, e.g., 
\cite{Slinko2004, XiaConitzer2008, MPR2012}).
If, on the other hand, the voters are not independent, say under the condition that all voting situations are equiprobable, then the average size of the minimal manipulating coalition may have order $n$ (\cite{Slinko2006}). 

The early literature on threats and counter-threats should be mentioned here although it is not on coalitional manipulability but rather on coalitional contr-manipulability. \citep{Pattanaik1976a,Pattanaik1976b,Barbera1980} assessed how coalitions (which could potentially be made up of any number of different voter-types) could respond to threats by individuals to manipulate; it may be interesting to look at situations in which one voter-type threatens to manipulate and another type then issues a counter-threat.

The main drawback of the concepts of coalitional manipulability and contr-manipulability is that they assume that coalitions can somehow be formed without specifying the mechanism.  Those voters must find each other, agree to form a coalition, calculate possible manipulations, negotiate which one to implement. Negotiation may however be difficult, say in our Example coalition consisting of voter 1 and voter 2 can manipulate the election, however, there is no obvious way for them to decide which course of action to take. This model, however, makes sense when there is an external actor, say a briber, who directs the voters she bribed. In contrast our voters assumed to be independent and have to take their own decisions which makes them players in a voting manipulation game. We also assume that only voters with identical preferences may try to a certain extent to coordinate their actions. 

The main idea of overshooting came to authors observing the behaviour of voters during 2005 general election in New Zealand conducted under the system of proportional representation \cite{SlinkoWhite2010}. There we gave examples of how safe and unsafe strategic votes could arise in that framework. To the best of our knowledge, the distinction between safe and unsafe strategic votes, and the concepts of strategic over/undershooting, were first introduced (in the context of strategic voting under proportional representation) in this paper.  ~\cite{ParikhPacuit2005} also used the expression `safe vote', but to mean a vote that is strategically superior to an abstention. 

 ~\cite{Batto2008} and~\cite{ElyBaliga2012} have described multi-winner elections in which strategic desertion of leading candidates appears to have gone too far, and strategic voters have overshot. 

There is an important relationship between the quality and quantity of information a voter has, and can receive and send via their communication networks, and a voter's ability to identify and respond to incentives to vote strategically.   
~\cite{ChopraPacuitParikh2004} made clear the importance of the relationship.  They stressed that the Gibbard-Satterthwaite theorem becomes `effective' only when voters are able to acquire a certain amount of information.   Although~\cite{ChopraPacuitParikh2004} made this point particularly thoroughly, the basic idea goes back much further -~\cite{DummettFarquharson1961}, for instance, wrote that ``[t]he only hypothesis which would make the assumption of uniformly sincere voting plausible would be the absurdly restrictive one that no voter had any knowledge, before or during the voting, of the preference scales of others''. 

In response to our preprints, certain follow-on research has already begun.  On the issue of complexity,~\cite{HazonElkind2010} and~\cite{IanovskiYuElkindWilson2011} presented polynomial time algorithms for finding a safe strategic vote under particular social choice rules.  And~\cite{WilsonReyhani2010} looked into the asymptotic probability of a safe manipulation existing for a given scoring rule.





\section{The model and the results}
\label{sec:model}

We consider the situation where a finite set of {\em voters} $[n]=\{1,2,\ldots, n\}$ are to choose a single alternative from a finite set of {\em alternatives} ${\cal A}$. Throughout the paper it will be assumed that $|{\cal A}|\ge 3$ and we will  denote the alternatives  $A, B, C, \ldots $.  Voters have preferences over alternatives represented by strict linear orders on $\cal A$. 
Voters whose preferences are identical will be referred to as being of the same {\em type}.  
When $|\mathcal{A}|=3$, a voter who prefers $X$ to $Y$ and $Y$ to $Z$, where $\{X,Y,Z\}=\{A,B,C\}$, will be referred to as an XYZ type; analogous language will be used when $|\mathcal{A}|>3$.  
A {\em profile} is a sequence of linear orders specifying one preference order for each voter.  A profile will often be represented as an $n$-tuple $R = (R_1, \ldots, R_n)$; further we will introduce an alternative notation.  

Given a set of alternatives $\cal A$ and a set of voters $[n]$, a {\em social choice rule} $F$ is a mapping from the set of all possible profiles to $\cal A$.  Voter $i$ is the {\em dictator} of $F$ if they would not desire to change the value of $F$ at any profile or, which is the same, $F$ always selects the most preferred alternative of voter $i$.  A social choice rule $F$ is {\em anonymous} if it does not pay attention to voters' labels, i.e., if $\pi$ is a permutationn of $(1,2,\ldots, n)$, then 
\[
F(\row Rn)=F(R_{\pi(1)},\ldots, R_{\pi(n)}).
\]  

Let $R$ be a profile, $V\subseteq [n]$ a set of voters of the same type, and $L$ a preference order over $\cal A$ other than the one representing the preferences of the voters in $V$.  Then $R_{-V}(L)$ shall denote the profile obtained from $R$ by replacing, for every $i \in V$, linear order $R_i$ with $L$, ceteris paribus.  $R_{-V}(L)$ can be read informally as ``the profile $R$, except that the preferences of all the voters in $V$ have been switched to $L$''.  If $X, Y \in \mathcal{A}$ and $V \subseteq [n]$ then $X \succ_V Y$ ($X \succeq_V Y$) will denote that every voter in $V$ ranks $X$ above (no lower than) $Y$. 

We are now ready to formulate

\begin{theorem}[Gibbard-Satterthwaite]
\label{GStheorem}
Suppose an onto, nondictatorial social choice rule $F$ is employed to choose one of at least three alternatives.  Then there exists a profile $R$, a linear order $L$, and a voter $i$, such that $F({R_{-\{i\}}}(L))\succ_{\{i\}}F(R)$.
\end{theorem}

\noindent Proofs can be found in the original papers~\cite{Gibbard1973} and~\cite{Satterthwaite1975}.  \par\medskip
For the next four definitions we fix a set of voters $[n]$, a set of alternatives $\cal A$, and a social choice rule $F$.  

\begin{definition}[An incentive to vote strategically]
Let $R$ be a profile, $i$ be a voter, and $V$ be the set of all voters with preferences identical to those of $i$ at $R$.  If there exists a linear order $L \ne R_i$ and a subset $V_1 \subseteq V$ containing $i$ such that
\[
F(R_{-V_1}(L)) \succ_V F(R)
\]
then we will say that, at $R$, voter~$i$ has {\em an incentive to vote strategically}.
\end{definition}


A voter with an incentive to vote strategically cannot necessarily change the social choice by themselves.  They may, however, hope that if they do cast the strategic vote, they will be joined in that act by the `right' set of like-minded fellow voters.  In some circumstances, if the `wrong' set of like-minded voters act on a shared incentive they can inadvertently make the outcome worse rather than better (from their perspective, of course).  In such circumstances we call the strategic vote in question unsafe. 

\begin{definition}[Unsafe strategic vote]
Suppose that at the profile $R$, voter $i$ has an incentive to strategically vote $L \ne R_i$.  Define $V \subseteq [n]$ to be the set of all voters with preferences identical to those of $i$ at $R$.  The strategic vote $L$ is unsafe for voter $i$ if there exists $V_2 \subseteq V$ such that
 $i \in V_2$, all the voters in $V_2$ have an incentive to strategically vote $L$, but
$F(R) \succ_V F(R_{-V_2}(L))$.
\end{definition}

A strategic vote that is not unsafe will be referred to as safe. During our examples and proofs we will refer to a particular kind of safe manipulation as an {\em escape}.  Let $R$ be a profile, and suppose $F$ maps $R$ to (say) $A$.  Let $L$ be one particular order that has $A$ last, and let $V$ be the entire set of all voters who at $R$ are of type $L$.  If any voter in $V$ has an incentive to vote strategically then voters in $V$ will be said to be able to escape at (or from) $R$.  

The most readily identifiable unsafe strategic votes arise when it is possible to get a result worse than the status quo when either `too many' or `too few' voters of a particular type act on a shared incentive to vote strategically.  

\begin{definition}[Strategic overshooting]
Let $R$ be a profile, $i$ a voter, and $L$ a linear order over $\mathcal{A}$ other than $R_i$.  Define $V$ to be the set of all voters with preferences identical to those of $i$ at $R$.   If we can find two subsets $V_2 \subsetneq V_1$ of $ V$, both containing $i$, such that every voter in $V_2$ has an incentive to strategically vote $L$ but
\[
F(R_{-V_1}(L)) \succ_V F(R) \succ_V F(R_{-V_2}(L)),
\]
then we say that at $R$ voter $i$ can {\em strategically overshoot} by voting $L$.
\end{definition}

\noindent Strategic overshooting occurs when too many voters act strategically.  Strategic undershooting occurs when too few do.  More generally, strategic undershooting occurs when the inclusion relation between $V_1$ and $V_2$ is reversed.

\begin{definition}[Strategic undershooting]
Let $R$ be a profile, $i$ be a voter, and $L$ be a linear order over $\mathcal{A}$ other than $R_i$.  Define $V$ to be the set of all voters with preferences identical to those of $i$ at $R$.  If we can find two subsets $V_2 \subsetneq V_1$ of $ V$, both containing $i$, such that every voter in $V_2$ has an incentive to strategically vote $L$ but
\[
F(R_{-V_1}(L)) \succ_V F(R) \succ_V F(R_{-V_2}(L)),
\]
then at $R$ voter $i$ can {\em strategically undershoot} by voting $L$.
\end{definition}


Under anonymous social choice rules, unsafe strategic votes arise when and only when strategic voters face the prospect of over- or undershooting.  Under non-anonymous social choice rules it is not necessarily the case that voters with an incentive to cast an unsafe strategic vote can strategically over- or undershoot. Since such rules are exotic we omit this example.

The distinction between safe and unsafe strategic votes leads naturally to the idea of a choice rule being safely manipulable.

\begin{definition}[Safe manipulability]
A social choice rule is safely manipulable if, under it, a voter can have an incentive to cast a safe strategic vote.
\end{definition}

The concepts we have formally defined enable us to broaden the Gibbard-Satterthwaite framework in the direction sought.  In Section~\ref{sec:examples} we demonstrate that our concepts are merited precisely because the situations to which they apply are not unusual.  In the subsequent sections we will prove the following theorem.  

\begin{theorem}
\label{ASSWtheorem} 
Every onto and non-dictatorial social choice rule facing at least three alternatives is safely manipulable.
\end{theorem}

\noindent The proof of Theorem~\ref{ASSWtheorem} invokes the Gibbard-Satterthwaite Theorem, but it is not a straightforward extension. This theorem implies our main result which extends the  Gibbard-Satterthwaite theorem.

\begin{theorem}
\label{extendGS}
Suppose an onto, nondictatorial social choice rule $F$ is employed to choose one of at least three alternatives.  Then there exists a profile $R$, a linear order $L$, and a voter $i$, such that $F({R_{-\{i\}}}(L))\succ_{\{i\}}F(R)$, and $L$ is a safe strategic vote for $i$.
\end{theorem}



\section{Geometry of overshooting. Further examples}
\label{sec:examples}

In this section we aim to achieve the following. Firstly, we give a geometric interpretation of safe and unsafe manipulation (overshooting) for scoring rules. Then we give an example of undershooting. It is a bit more trickier to construct than examples of overshooting and we need at least four alternatives for this. Finally we will present a profile which is unsafely but not safely manipulable. 


\begin{example}[The Borda rule; safe and unsafe manipulations]
\label{Borda_geom}
Suppose 94 voters are choosing one of $A$, $B$, or $C$.  Let the table below give the distribution of sincere preferences.

\begin{center}
\begin{tabular}{lcccccc}
\arrayrulecolor{Gray}
Preference order & $ABC$ & $ACB$ & $BAC$ & $BCA$ & $CAB$ & $CBA$ \\
Number of voters & 17 & 15 & 18 & 16 & 14 & 14
\end{tabular}
\end{center}

\noindent  If all voters are sincere then $A$ will score 96, $B$ 99, and $C$ 87, and $B$ would win.  If between four and eight $ABC$ types vote $ACB$, ceteris paribus, $A$ would win.  If 10 or more $ABC$ types vote $ACB$, ceteris paribus, $C$ would win.  So the profile of sincere preferences is prone to unsafe strategic voting.  We can express these outcomes using the terminology introduced in Definition~2.  Let $R$ denote the profile of sincere preferences, $V$ the set of 17 $ABC$ types, and let $V_1 \subsetneq V_2 \subseteq V$ be such that $4 \le |V_1| \le 8$ and $10 \le |V_2|$.  Then (compare Definition~2) every voter in $V_2$ has an incentive to strategically vote $ACB$, and
\[
F(R_{-V_1}(ACB)) \succ_V F(R) \succ_V F(R_{-V_2}(ACB)).
\]   
The profile of sincere preferences is also prone to safe strategic voting (by voters of a different type): if 13 or more $ACB$ voters vote $CAB$, ceteris paribus, then $C$ (rather than their least favorite $B$) will win.
\end{example}

We will build on the ideas of~\cite{Saari1994} to illustrate this example geometrically. Firstly, we normalise the Borda scores of $A,B,C$ so that the sum of normalised scores is $1$.  Let $scn(X)$ denote the normalised score of alternative $X$.  Consider the three-dimensional simplex $S^2$ with vertices labeled $A,B,C$.  A ballot outcome can be represented by the point ${\bf x}$ of $ S^{2}$ for which $x_1=scn(A)$, $x_2=scn(B)$ and $x_3=scn(C)$, and where $x_1$, $x_2$, and $x_3$ are the lengths of perpendiculars dropped from ${\bf x}$ onto $BC$, $AC$, and $AB$, respectively (i.e., $x_1,x_2,x_3$ are the barycentric coordinates of ${\bf x}$):
\bigskip

\vspace{-5mm}
\begin{center}
\beginpicture
\setcoordinatesystem units
<0.7mm,0.7mm>
\setplotarea x from -60 to 60, y from -5 to 60
 {
\setlinear
\plot 0 0 60 0 30 51.9 0 0 /
\put{$A$} at -3 -3
\put{$B$} at 63 -3
\put{$C$} at 30 56
\put {$\bullet$} at 34 15
\put{ ${\bf x}$} at 34 20
\plot 34 0 34 15 47 23 /
\plot 34 15 16 27  /
\put{ $x_1$} at 42.5 16.5
\put{ $x_2$} at 22 19
\put{ $x_3$} at 29 8
}
\endpicture
\end{center}
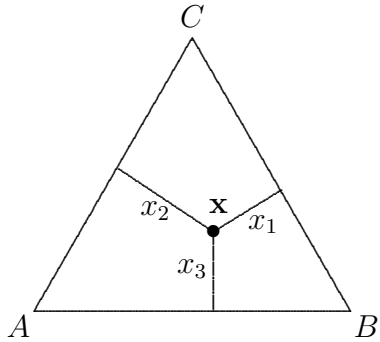
\centerline{\textbf{Figure 1}: Scores as homogeneous barycentric coordinates}\par\bigskip
\bigskip

\noindent The points of $S^2$ realisable as vectors of normalised Borda scores must lie within the region shaded in Figure~2 (below).  That shaded region is divided into three pentagons.  Whenever the vector of normalised scores falls into the pentagon closest to the vertice labelled $Y$, alternative $Y$ wins.

\vbox{
\begin{center}
\beginpicture 
\setcoordinatesystem units
<0.7mm,0.7mm>
\setplotarea x from -60 to 60, y from -5 to 60
\setlinear
\plot 0 0 60 0 30 51.9 0 0 /
\plot 30 0 30 17.3 45 26 /
\plot 30 17.3 15 26 /
\put{$A$} at -3 -3
\put{$B$} at 63 -3
\put{$C$} at 30 56
\setlinear
\plot 40 0  50 17.32 / 
\plot 40 34.64 20 34.64 /
\plot 10 17.32  20 0 /
\setshadegrid span <1.5pt>
\hshade 0 20 40 <,z,,>  17.32 10 50 /
\hshade 17.32 10 50 <,z,,>  34.64 20 40 /
\arrow <3mm> [.2,.5] from 31.5 11 to 24.0 24.5 
\arrow <3mm> [.2,.5] from 32.0 11 to 42.2 28 
\put {$\bullet$} at 31.75 11 
\put{${\bf x}$} at 32.75 7 
\endpicture
\end{center}
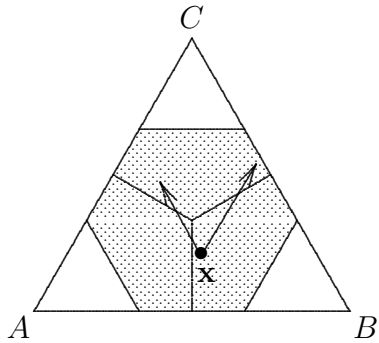
\centerline{\textbf{Figure 2}: Unsafe and safe manipulations}
}
\bigskip

\noindent We return to Example~\ref{Borda_geom}.  The vector of scores arising from the profile of sincere preferences is represented in Figure~2 by the point ${\bf x}$.  That point lies in the pentagon in which $B$ wins.  If some $ABC$ types vote $ACB$ then they move the outcome northwest, parallel to $BC$ (the score of $A$ remains unchanged), and into the region where $A$ wins.  If too many $ABC$ types vote $ACB$, the outcome moves all the way into the region where $C$ wins.  If, instead, all the $ABC$ types report their sincere preferences while some $ACB$ types vote $CAB$ then the outcome moves northeast from ${\bf x}$, parallel to $AC$ (the score of $B$ is unchanged), and, possibly, into the region where $C$ is the winning alternative.

To construct an example of strategic undershooting under a social choice rule, we need to have at least four alternatives present.

\begin{example}[The Borda rule; strategic undershooting]
Suppose $41$ voters are using the Borda rule to select one of five alternatives.  Let sincere preferences be distributed as follows.

\begin{center}
\begin{tabular}{lcccc}
\arrayrulecolor{Gray}
Preference order &  $ABCDE$ & $CEBAD$ & $EBCAD$ & $EDACB$   \\
Number of voters & 10 & 15 & 14 & 2  
\end{tabular}
\end{center}

\noindent   When all voters state their sincere preferences, $A$ scores 59, $B$ 102, $C$ 110, $D$ 30, and $E$ 109; $C$ wins.  If between two and six $ABCDE$ types vote $BADCE$, ceteris paribus, then $E$ wins.  If eight or more $ABCDE$ types vote $BADCE$, ceteris paribus, then $B$ wins.
\end{example}


\begin{example}[Anti-plurality/2-approval; a profile that is unsafely but not safely manipulable]
\noindent  Let the distribution of sincere preferences over the three alternatives on offer be as follows.

\begin{center}
\begin{tabular}{lcccccc}
\arrayrulecolor{Gray}
Preference order & $ABC$ & $ACB$ & $BAC$ & $BCA$ & $CAB$ & $CBA$ \\
Number of voters & 8 & 4 & 7 & 5 & 4 & 5
\end{tabular}
\end{center}

\noindent If none of the 33 voters are strategic then $A$ scores 23, $B$ 25, and $C$ 18, and $B$ wins.  At the profile of sincere preferences, only voters of type $ABC$ have an incentive to vote strategically.  If either 3 or 4 of them vote $ACB$, ceteris paribus, then $A$ will win.  If 6 or more of them vote $ACB$, ceteris paribus, $C$ will win.
\end{example}

\section{Proof of Theorem~\ref{ASSWtheorem}: the road map and the building blocks}
\label{sec:tools}

The proof of Theorem~\ref{ASSWtheorem} is split over three sections, so a road map may be helpful.  Here in Section~\ref{sec:tools} we collect technical statements that will be used frequently later on.  In Section~\ref{sec:threealternatives} we prove Theorem~\ref{ASSWtheorem} for the case of three alternatives.  In Lemma~\ref{tricky} we deal first  with the case of two voters,  then  the general case will be dealt with in Lemma~\ref{3Aweaklyu}.  
Section~\ref{sec:fouralternatives} looks at the situation when three or more alternatives are present.  At this stage of the proof we pay a special attention to completely agreed profiles, i.e., those profiles in which all linear orders coincide. There, we first  prove our theorem for the case that at least two alternatives are missing from the set of alternatives that can be elected at completely agreed profiles (Proposition~\ref{theorem2fromOS2}).  Two short Propositions~\ref{short3} and~\ref{ppr}, prepair us for the final thrust in the proof by induction dealing with two particular cases of Theorem~\ref{ASSWtheorem}.  Then, in a series of six short claims, we complete the proof of Theorem~\ref{ASSWtheorem}.  Section~\ref{sec:fouralternatives} finishes with a proof of Theorem~\ref{extendGS}.


We now need the following 
definition.  Let $F$ be a social choice rule, $R$ a profile, $L$ a preference order over the alternatives, and $V$ the entire set of voters having some particular preference order at $R$, common for all of them and different from $L$.  Then a subset $V_1 \subsetneq V$ will be classified as {\em L-inferior} if $$F(R_{-V}(L)) \succ_V F(R_{-V_1}(L)).$$  In other words, $V_1$ is $L$-inferior if voters from $V$ will be strictly better off if all of them switch to voting $L$ rather than only voters from $V_1$ do this.


\begin{proposition}
\label{inferiorproposition}
Let $F$ be a social choice rule.  Fix a  profile $R$ and a preference order  $L$, both over the set of alternatives $\cal A$.  Suppose $V$ is the entire set of voters who at $R$ have some preference order other than $L$.  If $V$ has an $L$-inferior subset then $F$ is safely manipulable.
\end{proposition}

\begin{proof}
Let $V_1$ be a maximal element of the set of $L$-inferior subsets of $V$ partially ordered by inclusion.  We claim that were $R_{-V_1}(L)$ the profile of sincere preferences then it would be safely manipulable by the voters in $V -  V_1$.  At the profile $R_{-V_1}(L)$, the voters in $V -  V_1$ are the sole voters present with their particular preferences.  If $\emptyset\ne V_2 \subseteq V -  V_1$, then $V_1\cup V_2$ is not $L$-inferior (as $V_1$ was maximal with this property) and one has
\[
F (  R_{-V_1}(L)_{-V_2}(L) ) = F(R_{-(V_1\cup V_2)}(L)) \succeq_V F(R_{-V}(L)) \succ_V F(R_{-V_1}(L)).
\]
This implies voters in $V -  V_1$,  at $R_{-V_1}(L)$, have incentive to strategically vote $L$ and can do it safely.
\end{proof}


\begin{proposition}
\label{endup}
Let $F$ be a social choice rule.  Suppose that, at a profile $R$, voter~$i$ has an incentive to strategically vote $L \ne R_i$.  Let $V$ be the entire set of voters with preferences identical to voter~$i$ at $R$.  If $F(R_{-V}(L)) \succeq_V F(R)$ then $F$ is safely manipulable. 
\end{proposition}

\begin{proof}
If $R$ is not safely manipulable, there exists a nonempty subset $V_1 \subsetneq V$ such that $F(R) \succ_V F(R_{-V_1}(L))$.  We  have $F(R_{-V}(L)) \succeq_V F(R) \succ_V F(R_{-V_1}(L))$.  So $V_1$ is an $L$-inferior subset of $V$, and by Proposition~\ref{inferiorproposition}, $F$ is safely manipulable.
\end{proof}


The following proposition will be of great help each time we need to show some social choice rule is safely manipulable; it allows not to consider cases of undershooting.

\begin{proposition}
\label{onlyundershooting}
Let $F$ be a social choice rule.  Suppose that at the profile $R$ voter $i$ has an incentive to cast the  strategic vote $L$ and it is unsafe.  Then either $F$ is safely manipulable or at some profile a voter can strategically overshoot.
\end{proposition}
\begin{proof}
Let $V = \{ j \mid R_j = R_i \}$.  Let $V_1 \subseteq V$ be any subset such that $i \in V_1$ and $F(R_{-V_1}(L)) \succ_V F(R)$. Such subset exists since $i$ has incentive to vote strategically.  Partially order the subsets of $V_1$ by inclusion. Let us consider the set of subsets $U\subset V_1$ such that $F(R) \succeq_V F(R_{-U}(L))$. Such set is nonempty since the empty set belongs to it.
Let $V_2$ be a maximal subset in this set. The subset  $V_2$, then must be a proper subset of $V_1$.  So we can find $j \in V_1 - V_2$, and for this voter, $$F((R_{-V_2}(L))_{-\{j\}}(L))=F(R_{-V_2 \cup \{ j \}}(L)) \succ_V F(R).$$  Now either voting $L$ is safe for $j$ at $R_{-V_2}(L)$ or it is unsafe; if the latter it must be because, at $R_{-V_2}(L)$, voter $j$ could strategically overshoot.
\end{proof}


We will call a social choice rule $F$ {\em antagonistic} if there exists a profile at which every voter ranks a particular alternative $X \in \cal A$ last, yet the value of $F$ at that profile is precisely $X$.

\begin{proposition}
\label{antagonisticproposition}
A non-constant antagonistic  social choice rule $F$  is safely manipulable.  
\end{proposition}

\begin{proof}
Let $R$ be the profile at which $A$ (without loss of generality) is last in every order $R_i$ but $F(R)=A$.  The rule $F$ is not constant, so let $Q$ be a profile such that $F(Q) \ne A$.  Take $R$ and for $i = 1, 2, \ldots$, one by one change $R_i$ to $Q_i$.  Were it a profile of sincere preferences, the last profile encountered for which $F$ does not take the value $A$ will clearly be safely manipulable. 
\end{proof}


The following construction, which reduces an arbitrary social choice rule $F$ to a two-voter rule, will be frequently used.  Let $V_1$ and $V_2$ be two non-empty non-intersecting subsets that partition the set of voters, i.e., $V_1\cup V_2=[n]$.  The value of the {\em two-voter social choice rule} $F_{V_1,V_2}$ at the two-voter profile ($R_1, R_2$) shall be the value of $F$ when all voters in $V_1$ report their preferences to be $R_1$ and all voters in $V_2$ report their preferences to be $R_2$.  We now have to aggregate voters for $F_{V_1,V_2}$. To distinguish them from those in the original full set $[n]$ they shall be denoted $V_1$ and $V_2$.

\begin{proposition}
\label{ifpartition}
Let $F$ be a social choice rule.  Let $V_1$ and $V_2$ form a non-trivial partition of $[n]$.  If $F_{V_1,V_2}$ is safely manipulable then so is $F$.
\end{proposition}

\begin{proof}
If $F_{V_1,V_2}$ is safely manipulable there exists a linear order $L$ and a two-voter profile $R=(R_1,R_2)$ such that, either $R_1 \ne R_2$ and
\begin{equation}
\label{halfway}
F_{V_1,V_2}(L,R_2)\succ_{\{V_1\}} F_{V_1,V_2}(R_1,R_2),
\end{equation}
or $R_1=R_2$ and
\begin{equation}
\label{alltheway}
F_{V_1,V_2}(L,R_2)\succeq_{\{V_1\}} F_{V_1,V_2}(R_1,R_2),\quad  F_{V_1,V_2}(L,L)\succeq_{\{V_1\}} F_{V_1,V_2}(R_1,R_2)
\end{equation}
with one of the two relations in (\ref{alltheway}) being strict.  Thus, at the two-voter profile $R$, voter $V_1$ can safely strategically vote $L$.  Let $Q$ be the $n$-voter profile for which $Q_v = R_1$ when $v \in V_1$ and $Q_v = R_2$ when $v \in V_2$. Without loss of generality suppose $1\in V_1$. At $Q$, voter~1 has an incentive to vote strategically; with that incentive in mind, we can finish by appealing to Proposition~\ref{endup} using (\ref{halfway}), when $R_1 \ne R_2$, or (\ref{alltheway}), when $R_1 = R_2$.
\end{proof}


A profile $R$ is {\em completely agreed} if all voters have identical preferences at $R$.  When the social choice rule $F$ is clearly fixed, $\mathcal{C}$ will denote the set of values that $F$ takes on completely agreed profiles.    A social choice rule is {\em weakly unanimous} if it selects every voter's favorite alternative at all completely agreed profiles (`weakly' because such a rule will not necessarily select $X$ when all voters report they rank $X$ first).


\section{Proof  of Theorem~\ref{ASSWtheorem}: three alternatives}
\label{sec:threealternatives}

In this section we always assume that $|\mathcal{A}|=3$. We will find it convenient to depict profiles over $\mathcal{A}$ as 6-tuples of sets: $R = (X_1,X_2,X_3,X_4,X_5,X_6)$, where the six sets in the sequence are the set of voters of types  $ABC$, $ACB$, $BAC$, $BCA$, $CAB$, and $CBA$, respectively. 


\begin{proposition}
\label{bac_bca_cba} 
Let $F$ be a social choice rule.  If, at a profile $R$, $ABC$ types can strategically overshoot by voting $BAC$, $BCA$, or $CBA$, then $F$ is safely manipulable.
\end{proposition}

\begin{proof}
Let $R$ be a profile and $L \in \{ BAC, BCA, CBA\}$ a linear order such that, if $V$ is the entire set of all voters with preferences $ABC$ at $R$, firstly
\[
 F(R)=B, \quad\text{and}\quad F(R_{-V}(L))=C
 \] 
and, secondly, 
\[
F(R_{-V_1}(L)) = A
\]
for some $V_1 \subset V$.  Then $ABC$ types, at $R$, can overshoot by voting $L$.  If $L=BAC$ then if $R_{-V}(L)$ were the profile of sincere preferences, $BAC$ voters could escape by voting $ABC$.  If $L$ equals $BCA$ then if $R_{-V_1}(L)$ were the profile of sincere preferences, $BCA$ voters could escape by voting $ABC$.  If $L$ equals $CBA$ then, similarly, if $R_{-V_1}(L)$ were the profile of sincere preferences, $CBA$ voters could escape by voting $ABC$.
\end{proof}

We now prove our theorem for the two-voter-three-alternative case.


\begin{lemma}
\label{tricky}
Let $F$ be an onto and non-dictatorial social choice rule.  If $n = 2$ then $F$ is safely manipulable.  If $n \ge 2$ and $V_1$ and $V_2$ form a non-trivial partition of $[n]$ for which $F_{V_1,V_2}$ is onto and non-dictatorial, then $F$ is safely manipulable.
\end{lemma}

\begin{proof} 
By Proposition~\ref{ifpartition} the second statement follows from the first.  We now prove the first statement.  By the Gibbard-Satterthwaite theorem, we can suppose (without loss of generality) voter $1$ can individually manipulate a profile $R$ with a vote of $L \ne R_1$.  If this manipulation is not safe then $R_1 = R_2$ and at $R$ both voters have an incentive to strategically vote $L$.  Assume (Proposition~\ref{onlyundershooting}) that the lack of safety arises from the prospect of an overshoot.  Without any loss of generality, let $R_1 = R_2 = ABC$.  Then
\[
B = F(ABC, ABC), \quad A = F(L, ABC), \quad \text{and} \quad C = F(L, L).
\]
By Proposition~\ref{bac_bca_cba} we only need to consider $L \in \{ACB,CAB\}$.
\bigskip

\noindent {\bf Case 1:} $L = ACB$.  It may be convenient to refer to the table below while reading the proof.  Rows represent voter~1's vote, columns voter~2's vote.  Cell entries indicate the value of $F$ at the relevant vote pairing.

\bigskip
\begin{center}
\begin{tabular}{cc|c|c|c|}
\arrayrulecolor{Gray}
& \multicolumn{1}{c}{} & \multicolumn{3}{c}{voter 2} \\
&  & $ABC$ & $ACB$ & $BAC$ \\
\cline{2-5}
\multirow{3}{*}{voter 1} & $ABC$ & $B$ & & \\
\cline{2-5}
& $ACB$ & $A$ & $C$ & \\
\cline{2-5}
& $CAB$ &  &  & \\
\cline{2-5}
\end{tabular}
\end{center}
\bigskip

\noindent  If $F(CAB, ACB) \ne C$ then voter~1 can safely manipulate from $(CAB, ACB)$ to $(ACB, ACB)$.  So suppose $F(CAB, ACB) = C$.  If $F(CAB, ABC) = B$ then voter~1 can escape from $(CAB, ABC)$ to $(ACB, ABC)$; if $F(CAB, ABC) = A$ then voter~2 can safely manipulate from $(CAB, ACB)$ to $(CAB, ABC)$.  So suppose $F(CAB, ABC) = C$.  Then $F(CAB, BAC) = C$, for if not, voter~2 can escape from $(CAB, ABC)$ to $(CAB, BAC)$.  Next consider $(ACB, BAC)$.  If $F(ACB, BAC) = B$ then voter~1 can escape from $(ACB, BAC)$ to $(CAB, BAC)$.  If $F(ACB, BAC) = C$ then voter~2 can escape from $(ACB, BAC)$ to $(ACB, ABC)$.  So let $F(ACB, BAC) = A$.  Then $F(ABC, BAC) = A$, otherwise voter~1 can safely manipulate from $(ABC, BAC)$ to $(ACB, BAC)$.  But now voter~2 can safely manipulate from $(ABC, BAC)$ to $(ABC, ABC)$.\par

\bigskip

\noindent {\bf Case 2:} $L = CAB$.

\bigskip
\begin{center}
\begin{tabular}{cc|c|c|}
\arrayrulecolor{Gray}
& \multicolumn{1}{c}{} & \multicolumn{2}{c}{voter 2} \\
&  & $ABC$ & $ACB$  \\
\cline{2-4}
\multirow{3}{*}{voter 1} & $ABC$ & $B$ & \\
\cline{2-4}
& $ACB$ & &  \\
\cline{2-4}
& $CAB$ & $A$  &  \\
\cline{2-4}
\end{tabular}
\end{center}
\bigskip

\noindent If $F(ACB, ABC) \ne A$ then voter~1 can safely manipulate $(ACB, ABC)$ by voting $CAB$.  So let $F(ACB, ABC) = A$.  If $F(CAB, ACB) \ne A$ then voter~2 can safely manipulate $(CAB, ACB)$ by voting $ABC$.  So assume $F(CAB, ACB) = A$.  If $F(ABC, ACB) \ne A$ then voter~1 can safely manipulate $(ABC, ACB)$ by voting $CAB$.  So let $F(ABC, ACB) = A$.  If $F(ACB, ACB) = A$ or $B$ then a safe manipulation is possible at $(ABC, ABC)$.  If $F(ACB, ACB) = C$ then voter~1 can safely manipulate from $(CAB, ACB)$ to $(ACB, ACB)$.  
\end{proof}


\begin{lemma}
\label{3Anotweaklyu}
A social choice rule $F$ that is  onto but not weakly unanimous is safely manipulable. 
\end{lemma}

\begin{proof}
Assume $F$ isn't antagonistic (Proposition~\ref{antagonisticproposition}).  In this case there will be a profile $S$ at which every voter reports $ABC$ (say) but $F(S)$ is $B$ rather than $A$.  If there is any completely agreed profile $Q$ such that $F(Q) = A$ then at $S$ all voters have an incentive to vote $Q_1$; if this strategic vote isn't safe then we can apply Proposition~\ref{endup}.  So assume $A \not\in \mathcal{C}$.  Given $F$ is onto we can find a profile 
\[
R^0 = (X_1, X_2, X_3, X_4, X_5, X_6)
\]
such that $F(R^0) = A$.  Let
\[
R^1 = R^0_{-X_3}(ABC) = (X_1\cup X_3, X_2, \emptyset, X_4, X_5, X_6).
\]
If $F(R^1) = B$ then at $R^0$ voters of type $BAC$ have an incentive to vote $ABC$, and this strategic vote is such that we may apply Proposition~\ref{endup} to conclude $F$ is safely manipulable.  If $F(R^1) = C$ then at $R^1$ some $ABC$ types can escape by voting $BAC$.  This leaves us with the case that $F(R^1) = A$, which we now assume.  Consider
\[
R^2 = R^1_{-X_4}(ABC) = (X_1\cup X_3 \cup X_4, X_2, \emptyset, \emptyset, X_5, X_6).
\]
If $F(R^2) \ne A$ then at $R^1$ voters of type $BCA$ can escape by voting $ABC$.  So assume $F(R^2) = A$.  Let
\[
R^3 = R^2_{-X_5}(ACB) = (X_1\cup X_3 \cup X_4, X_2\cup X_5, \emptyset, \emptyset, \emptyset, X_6).
\]
If $F(R^3) = B$ then at $R^3$ some voters of type $ACB$ can escape by voting $CAB$.  If $F(R^3) = C$ then at $R^2$ voters of type $CAB$ have an incentive to vote $ABC$, and this strategic vote is such that we may apply Proposition~\ref{endup} to conclude $F$ is safely manipulable.  So assume $F(R^3) = A$, and consider
\[
R^4 = R^3_{-X_6}(ABC) = (X_1\cup X_3 \cup X_4 \cup X_6, X_2\cup X_5, \emptyset, \emptyset, \emptyset, \emptyset)
\]
If $F(R^4) \ne A$ then at $R^3$ voters of type $CBA$ can escape by voting $ABC$.  So assume $F(R^2) = A$.  

We are left in a situation in which $F_{X_1\cup X_3 \cup X_4 \cup X_6, X_2\cup X_5}$ is well-defined; moreover, it maps $(ABC,ACB)$ to $A$, $(ABC,ABC)$ to $B$, and $(ACB,ACB)$ to $C$, and so Lemma~\ref{tricky} applies.

\end{proof}


\begin{lemma}
\label{3Aweaklyu}
If $F$ is an onto and nondictatorial social choice rule then it is safely manipulable.
\end{lemma}

\begin{proof}
Throughout this proof we  may assume that three or more voters are present. Lemma~\ref{3Anotweaklyu} allows us to assume that $F$ is weakly unanimous.

Due to the Gibbard-Satterthwaite theorem it suffices to show that $F$ being unsafely manipulable implies $F$ is also safely manipulable.  By Proposition~\ref{onlyundershooting}, and without any loss of generality, suppose that some $ABC$ types may strategically overshoot at a profile $R$.  By Proposition~\ref{bac_bca_cba} we may assume that they may overshoot voting $ACB$ or $CAB$.\par\medskip

\noindent {\bf Case 1:} Suppose that, at $R$, some $ABC$ types can strategically overshoot by voting $ACB$.  Let $V$ be the set of $ABC$ types at $R$.  Then $F(R)=B$ and there must exist some $V_1 \subset V$ such that $F(R_{-V_1}(ACB)) = A$ while $F(R_{-V}(ACB)) = C$.  Let
\[
R^1 = R_{-V_1}(ACB) = (X_1, X_2, X_3, X_4, X_5, X_6).
\]
\noindent The intent now is to either directly show that $F$ is safely manipulable or to demonstrate that the two-voter social choice rule generated by $X_1\cup X_2\cup X_3$ and $X_4\cup X_5\cup X_6$ is well-defined, onto, and nondictatorial.  This will imply safe manipulability by Proposition~\ref{ifpartition}.
We know that $F(R^1) = A$, and that both $X_1$ and $X_2$ are not empty (and therefore that $X_1 \cup X_2 \cup X_3 \ne \emptyset$).  Let
\[
R^2 = R_{-V}(ACB) = (\emptyset, X_1\cup X_2, X_3, X_4, X_5, X_6).
\]
We know  $F(R^2) = C$.  Next let 
\[
R^3 = (X_1\cup X_2, \emptyset, X_3, X_4, X_5, X_6).
\]
 If $F(R^3) = A$ or $C$ then consider the manipulation of $R^2$ by some ACB types voting ABC; if this is unsafe we may then use Proposition~\ref{endup} to deduce $F$ is safely manipulable.
So suppose $F(R^3) = B$.  Let
\begin{align*}
R^4 &= (X_1, X_2, X_3, \emptyset, X_4\cup X_5, X_6),
\end{align*}
If $F(R^4) \ne A$ and $X_4\ne\emptyset$, then $BCA$ types can escape (to $R^4$) at $R^1$.  So suppose $F(R^4) = A$ (if $X_4 = \emptyset$ this is immediate as  then  $R^4=R^1$). Let us consider now
\begin{align*}
R^5 &= (X_1, X_2, X_3, \emptyset, X_4\cup X_5\cup X_6, \emptyset ).
\end{align*}
 If $F(R^5) \ne A$ and $X_6\ne\emptyset$, then $CBA$ types can escape (to $R^5$) at $R^4$.  So suppose $F(R^5) = A$  (if $X_6 = \emptyset$ this is immediate as then  $R^5=R^1$).  
 
Let
\begin{align*}
R^6 &= (X_1\cup X_3, X_2, \emptyset, \emptyset, X_4\cup X_5\cup X_6, \emptyset ).
\end{align*} 
Suppose, for now, $X_3 \ne \emptyset$.  If $F(R^6) = B$ then consider the manipulation of $R^5$ by $BAC$ types voting $ABC$; if this is unsafe we may then use Proposition~\ref{endup} to deduce $F$ is safely manipulable.  If $F(R^6) = C$ then some $ABC$ types can escape from $R^6$ (to $R^5$).  So suppose $F(R^6) = A$  (if $X_3 = \emptyset$ this is immediate).  

Consider now 
\begin{align*}
R^7 &= (X_1\cup X_2\cup X_3, \emptyset, \emptyset, \emptyset, X_4\cup X_5\cup X_6, \emptyset ).
\end{align*}
If $X_2 = \emptyset$, then $F(R^7) = F(R^6)= A \ne C$.  If $X_2 \ne \emptyset$ and $F(R^7) = C$, some $ABC$ types can escape $R^7$ (to $R^6$).  So suppose $F(R^7) \ne C$.  This implies that in the event $X_4 \cup X_5 \cup X_6 \ne \emptyset$, the second voter is not a dictator for $F_{X_1 \cup X_2 \cup X_3 , X_4 \cup X_5 \cup X_6}$.

We now show (assuming $F$ is not safely manipulable) that $X_4 \cup X_5 \cup X_6 \ne \emptyset$.  We will then show (again assuming $F$ is not safely manipulable) the first voter is not a dictator for $F_{X_1\cup X_2\cup X_3, X_4\cup X_5\cup X_6}$.  Given that $F$ is weakly unanimous, hence onto, we will then have enough to use Proposition~\ref{ifpartition}.  We need just three more profiles:
\begin{align*}
R^8 &= (\emptyset, X_1\cup X_2\cup X_3, \emptyset, X_4, X_5, X_6),\\
R^9 &= (\emptyset, X_1\cup X_2\cup X_3, \emptyset, \emptyset, X_5, X_4\cup X_6),\\
R^{10} &= (\emptyset, X_1\cup X_2\cup X_3, \emptyset, \emptyset, \emptyset, X_4\cup X_5\cup X_6).
\end{align*}
\noindent
If $X_3 = \emptyset$, then $F(R^8) =F(R^2)= C$.  Suppose $X_3 \ne \emptyset$.  If $F(R^8) \ne C$, then $BAC$ types can escape from $R^2$ (to $R^8$).  So assume $F(R^8) = C$.  Since $F(R^7) \ne C$, by the weak unanimity of $F$, this implies $X_4\cup X_5\cup X_6 \ne \emptyset$.  

If $X_4 = \emptyset$ then $F(R^9) =F(R^8) = C$.  Suppose $X_4 \ne \emptyset$.  If $F(R^9) = A$, then some $CBA$ types can escape from $R^9$ (to $R^8$).  So (regardless of whether $X_4$ is empty or not) let $F(R^9) \ne A$.  If $X_5 = \emptyset$ then $F(R^{10}) \ne A$.  Suppose $X_5 \ne \emptyset$.  If $F(R^{10}) = A$, then some $CBA$ types can escape from $R^{10}$ (to $R^9$).  So let $F(R^{10}) \ne A$.  This implies that the first voter is not a dictator for $F_{X_1\cup X_2\cup X_3, X_4\cup X_5\cup X_6}$. This function inherits weak unanimity from $F$ and hence is onto.  Then by Proposition~\ref{ifpartition}, $F$ is safely manipulable.\par\medskip


\noindent {\bf Case 2:} overshooting by voting CAB.  Suppose that, at $R$, some $ABC$ types can strategically overshoot by voting $CAB$.  This implies $F(R) = B$.  Let $V$ be the set of $ABC$ types at $R$.  There must exist some $V_1 \subset V$ such that $F(R$$_{-V_1}(CAB)) = A$.  Let
\[
R^1 = R_{-V_1}(CAB) = (X_1, X_2, X_3, X_4, X_5, X_6).
\]
\noindent The intent now is to either directly show that $F$ is safely manipulable or to demonstrate that the two-voter social choice rule generated by $X_1\cup X_2\cup X_5$ and $X_3\cup X_4\cup X_6$ is onto, and nondictatorial.  We know $F(R^1) = A$, and $X_1, X_5 \ne \emptyset$ (and hence $X_1 \cup X_2 \cup X_5 \ne \emptyset$).  Let
\[
R^2 = R_{-V}(CAB) = (\emptyset, X_2, X_3, X_4, X_1\cup X_5, X_6).
\]
\noindent If $F(R^2) = A$ or $B$ then we may apply Proposition~\ref{endup} to deduce $F$ is safely manipulable.  So suppose $F(R^2) = C$.  Next let 
\[
R^3 = (X_1\cup X_5, X_2, X_3, X_4, \emptyset, X_6).
\]
\noindent If there are no $CAB$ types present at $R$ then $R = R^3$, and $F(R^3) = B$.  Now suppose that there are $CAB$ types present at $R$, and $F(R^3) \ne B$; given $F(R) = B$, $CAB$ types are then capable of escaping from $R$ to $R^3$.  So let us suppose $F(R^3) = B$.  Let
\begin{align*}
R^4 &= (X_1\cup X_2\cup X_5, \emptyset, X_3, X_4, \emptyset, X_6),\\
R^5 &= (X_1\cup X_2\cup X_5, \emptyset, \emptyset, X_4, \emptyset, X_3\cup X_6),\\
R^6 &= (X_1\cup X_2\cup X_5, \emptyset, \emptyset, \emptyset, \emptyset, X_3\cup X_4\cup X_6).
\end{align*}
\noindent
If $X_2= \emptyset$, then $F(R^4)=F(R^3) = B$.  Suppose $X_2 \ne \emptyset$.  If $F(R^4) \ne B$ then $ACB$ types can escape from $R^3$ to $R^4$.  So let $F(R^4) = B$. We note that by weak unanimity this implies $X_3\cup X_4\cup X_6 \ne \emptyset$.  If $X_3 = \emptyset$, then $F(R^5)=F(R^4) = B$.  Suppose $X_3 \ne \emptyset$; if $F(R^5) = A$ then some $CBA$ types can escape from $R^5$ to $R^4$.  So, regardless of whether $X_3$ is empty or not, to proceed we let $F(R^5) \ne A$.  If $X_4 = \emptyset$, $F(R^6) \ne A$.  If $X_4 \ne \emptyset$ then in the event $F(R^6) = A$, some $CBA$ types can escape from $R^6$ to $R^5$.  So  $F(R^6) \ne A$.  This implies the first voter is not a dictator for $F_{X_1\cup X_2\cup X_5, X_3\cup X_4\cup X_6}$.

It remains to show that the second voter is not a dictator for $F_{X_1\cup X_2\cup X_5, X_3\cup X_4\cup X_6}$.  For this purpose we will need four more profiles:
\begin{align*}
R^7 &= (X_1, X_2, X_3\cup X_4, \emptyset, X_5, X_6),\\
R^8 &= (X_1, X_2, X_3\cup X_4\cup X_6, \emptyset, X_5, \emptyset),\\
R^ 9 &= (\emptyset, X_2, X_3\cup X_4\cup X_6, \emptyset, X_1\cup X_5, \emptyset),\\
R^{10} &= (\emptyset, \emptyset, X_3\cup X_4\cup X_6, \emptyset, X_1\cup X_2\cup X_5, \emptyset).
\end{align*}
\noindent If $X_4 = \emptyset$, then $F(R^7) =F(R^1)= A$.  Suppose $X_4 \ne \emptyset$.  If $F(R^7) \ne A$ then $BCA$ types can escape from $R^1$ to $R^7$.  So let $F(R^7) = A$.  If $X_6 = \emptyset$, then $F(R^8) =F(R^7)= A$.  Suppose $X_6 \ne \emptyset$.  If $F(R^8) \ne A$ then $CBA$ types can escape from $R^7$ to $R^8$.  So let $F(R^8) = A$.  If $F(R^9) = B$ then some $CAB$ types can escape from $R^9$ to $R^8$.  So let $F(R^9) \ne B$.  If $X_2 = \emptyset$, then $F(R^{10})=F(R^9) \ne B$.  Suppose $X_2 \ne \emptyset$.  If $F(R^{10}) = B$ then some $CAB$ types can escape from $R^{10}$ to $R^9$.  So let $F(R^{10}) \ne B$.  But then the second voter is not a dictator for $F_{X_1\cup X_2\cup X_5, X_3\cup X_4\cup X_6}$.

The rule $F_{X_1\cup X_2\cup X_5, X_3\cup X_4\cup X_6}$ inherits weak unanimity from $F$.  Hence this rule is onto.  Then by Proposition~\ref{ifpartition}, $F$ is safely manipulable.
\end{proof}


\section{Proof  of Theorem~\ref{ASSWtheorem}: three or more alternatives}
\label{sec:fouralternatives}


Proposition~\ref{theorem2fromOS2} shows Theorem~\ref{ASSWtheorem} holds when $|\mathcal{C}| < |\mathcal{A}| - 1$.  The two subsequent short propositions set up a proof-by-induction of the whole thing.


\begin{proposition}
\label{theorem2fromOS2}
Let $F$ be an onto social choice rule.  If at least two alternatives are missing from $\mathcal{C}$ then $F$ is safely manipulable.
\end{proposition}

\begin{proof}
Suppose that neither $A$ nor $B$ belong to $\mathcal{C}$.  Let $L^{AB}$ be a fixed but otherwise arbitrary linear order of the alternatives that has $A$ first and $B$ second.  Let $L^{BA}$ be the linear order formed by taking $L^{AB}$ and reversing the spots of $A$ and $B$, ceteris paribus.  Let $m \ge 2$ be the minimum possible number of voter types present when the value of $F$ is in the set $\{ A,B \}$.  Let $\cal S$ denote the set of profiles that have exactly $m$ voter types present and are mapped by $F$ to either $A$ or $B$.\par

Firstly, suppose no profile in $\cal S$ has an $L^{AB}$ type voter present.  Pick $R \in \cal S$.  Let $V$ be the entire set of voters having the preference order $R_1$ at $R$.  Consider the profile $R_{-V}(L^{AB})$; $F$ cannot map this profile to $A$ or $B$ because it has $m$ voter types present, and one of those types is $L^{AB}$.  However,
\begin{equation*}
F(R) = F((R_{-V}(L^{AB}))_{-V}(R_1)) \in \{A,B\} 
\end{equation*}
and so, at $R_{-V}(L^{AB})$, a voter of type $L^{AB}$ has an incentive to strategically vote $R_1$.  If such a strategic vote would be unsafe then we may apply Proposition~\ref{endup} to deduce $F$ is safely manipulable.  In the event that no $R \in \cal S$ has $L^{BA}$ type voters, the analysis proceeds similarly.\par

Secondly, suppose that some profile in $\cal S$ has an $L^{AB}$ type voter present, another has an $L^{BA}$ type voter present, but no profile in $\cal S$ has voters of both types present.  Let $R \in \cal S$ have $L^{AB}$ types present.  Since $R$ cannot be completely agreed there must be an $i$ such that $R_i \ne L^{AB}$.  Let $V = \{ j | R_j = R_i \}$.  Since $F(R_{-V}(L^{BA}))\notin \{A,B\}$ (the profile $R_{-V}(L^{BA})$ features both $L^{AB}$ and $L^{BA}$ types), the profile $R_{-V}(L^{BA})$ is manipulable by $L^{BA}$ voters voting $R_1$.  If this manipulation is unsafe, we can apply Proposition~\ref{endup}. 

Thirdly and finally, suppose $R \in \cal S$ has both $L^{AB}$ and $L^{BA}$ types present.  Let $U$ and $W$ be, respectively, the set of all voters with preferences $L^{AB}$ and $L^{BA}$ at $R$.  Then $F(R_{-U}(L^{BA}))\notin \{A,B\}$ and $F(R_{-W}(L^{AB}))\notin \{A,B\}$ as both the relevant profiles have only $m - 1$ types present.  Define a new relation on $\cal A$ as follows: $\succeq_{U\& W}$ if and only if $\succeq_U$ and $\succeq_W$.  Without loss of generality we assume 
\begin{equation*}
F(R_{-U}(L^{BA})) \succeq_{U\& W} F(R_{-W}(L^{AB})).
\end{equation*}
Now
\begin{equation*}
F(R) = F((R_{-W}(L^{AB}))_{-W}(L^{BA})) \succ_{U\& W} F(R_{-W}(L^{AB}))
\end{equation*}
because $F(R)\in \{A,B\}$ and $F(R_{-W}(L^{AB}))\notin \{A,B\}$.  So voters in $W \subset U\cup W$ can manipulate at $R_{-W}(L^{AB})$ (by insincerely voting $L^{BA}$ rather than sincerely voting $L^{AB}$).  If this manipulation is safe we are done; if not then notice that the manipulation is such that we may apply Proposition~\ref{endup}.
\end{proof}


\begin{proposition}
\label{short3} 
Let $F$ be a social choice rule.  If $Y \in \mathcal{C}$ but $Y$ is not in the range of $F_{-Z}$ for some $Z \ne Y$ then $F$ is safely manipulable.
\end{proposition}

\begin{proof}
Let $R$ be a unanimous profile at which all voters rank $Y$ first and $Z$ last, but $F(R) \ne Y$.  Given $Y$ is in $\mathcal{C}$, there is a completely agreed profile $Q$ mapped by $F$ to $Y$.  At $R$ every voter $i$ has an incentive to switch from $R_i$ to $Q_1$.  If this strategic vote is unsafe, we can apply Proposition~\ref{endup}.
\end{proof}


Many parts of the remainder of our proof will utilise `subrules' derived from the main rule $F$.  Given any alternative $X$, the subrule $F_{-X}$ is designed to pick alternatives from ${\cal A}-\{X\}$, and operates as follows.  Let $R$ be an arbitrary profile of preferences over the set ${\cal A} - \{X\}$.  Let $R'$ be the profile of preferences over the original set of alternatives $\cal A$ formed by appending an $X$ to the bottom of every preference order in $R$.  Then the value of $F_{-X}$ at $R$ shall be the value of $F$ at $R'$.  Provided $F$ is not antagonistic, $F_{-X}$ will not select $X$ at $R'$ and will therefore be sensibly defined.   We will say $F_{-X}$ is {\em proper} if it is non-dictatorial and has image ${\cal A} - \{X\}$.


\begin{proposition}
\label{ppr}
Let $F$ be a social choice rule.  If there is an $A \in \mathcal{A}$ for which $F_{-A}$ is proper and safely manipulable then $F$ itself is safely manipulable.
\end{proposition}

\begin{proof}
Suppose that when the proper rule $F_{-A}$ is in operation, voter $i$ has an incentive to safely strategically vote $L$ at the profile $R$ on ${\cal A}-\{A\}$.  Let $V$ be the set of indices $j$ such that $R_j=R_i$.  Then
\[
F_{-A}(R_{-U}(L))\succ_VF_{-A}(R)
\]
for one particular subset $U\subseteq V$ and $F_{-A}(R_{-W}(L))\succeq_VF_{-A}(R)$  for every subset ${W\subseteq V}$.  Let $\bar{R}=(\bar{R}_1,\ldots,\bar{R}_n)$ be the profile of preferences over $\cal A$ formed by appending an $A$ to the bottom of every preference order $R_i$ in $R$.  Let $\bar{L}$ be $L$ with $A$ appended to the bottom.  Then $V=\{j\mid \bar{R}_j=\bar{R}_i\}$.  Furthermore, 
\[
F(\bar{R}_{-U}(\bar{L})) =   F_{-A}(R_{-U}(L))\succ_VF_{-A}(R) =   F(\bar{R})
\]
and for any $W\subseteq V$
\[
F(\bar{R}_{-W}(\bar{L}))=F_{-A}(R_{-W}(L))\succeq_VF_{-A}(R)=F(\bar{R}).
\]
Hence voter~$i$ can safely manipulate at $\bar{R}$ with a vote of $\bar{L}$.
\end{proof}


We can now prove Theorem~\ref{ASSWtheorem}:

\begin{proof}
Lemma~\ref{3Aweaklyu} provides a base case for an inductive proof.  Assume $|\mathcal{A}| \ge 4$ and the statement holds when the number of alternatives is $|\mathcal{A}|-1$.  Assume $F$ is not anatagonistic (Proposition~\ref{antagonisticproposition}).  Assume $|\mathcal{C}| \ge |\mathcal{A}|-1$ (Proposition~\ref{theorem2fromOS2}).  Find an alternative $A$ as follows: if $|\mathcal{C}| = |\mathcal{A}| - 1$, let $A$ be the single alternative in $\mathcal{A} - \mathcal{C}$; if $|\mathcal{C}| = |\mathcal{A}|$ let $A$ be an arbitrary alternative; either way, we note, $\mathcal{A} - \{A \} \subseteq \mathcal{C}$.  The rule $F_{-A}$ is either dictatorial or it is not.  Suppose it is not: if $F_{-A}$ is proper then (by the induction hypothesis) Proposition~\ref{ppr} applies; if $F_{-A}$ isn't proper then there is some $Y \in \mathcal{C}$ ($Y \ne A$) it cannot reach, and Proposition~\ref{short3} applies.  Suppose $F_{-A}$ is dictatorial: then the following chain of claims shows $F$ has to be safely manipulable.   
\end{proof}


Note that the alternative $A$ has been fixed.  For convenience, let voter 1 be the dictator of $F_{-A}$.  In all of the following statements and proofs, $X$ will represent an alternative different from $A$, and $V_1$ and $V_2$ will denote $\{ 1 \}$ and $[n] - \{ 1 \}$ respectively.  Recall Proposition~\ref{ifpartition} showed that if $F_{V_1,V_2}$ is safely manipulable then so is $F$.


\begin{myclaim}
\label{myclaim1}
If $F_{-X}$ is non-dictatorial and can reach $A$ then $F$ is safely manipulable.
\end{myclaim}
\begin{proof}
If $F_{-X}$ is proper then Proposition~\ref{ppr} applies, otherwise there is some $Y \in \mathcal{C}$ ($Y \ne A$) it cannot reach, and we can appeal to Proposition~\ref{short3}.
\end{proof}


\begin{myclaim}
\label{myclaim2}
Either $F$ is safely manipulable or $F_{ V_1 , V_2 }$ returns $X$ whenever voter 1 ranks $X$ first.
\end{myclaim}

\begin{proof}
We try to avoid concluding $F$ is safely manipulable.

Voter 1 dictates $F_{-A}$, so $F_{ V_1 , V_2 }$ maps $( X\cdots A , \cdots XA )$ to $X$.  If that profile is not safely manipulable by voter $V_2$ then $F_{ V_1 , V_2 }( X\cdots A , \cdots AX ) \in \{ X, A \}$.

Suppose $( X\cdots A , \cdots AX )$ is mapped to $A$ rather than $X$.  If voter $V_1$ cannot safely manipulate that latter profile then $F_{ V_1 , V_2 }( \cdots AX , \cdots AX ) = A$, and Claim~\ref{myclaim1} applies.

So say $F_{ V_1 , V_2 }( X\cdots A , \cdots AX ) = X$.  If voter $V_2$ has no safe manipulations, $F_{ V_1 , V_2 }( X\cdots A , L ) = X$ for all $L$.  And if neither does voter $V_1$, $F_{ V_1 , V_2 }( X\cdots , L ) = X$ for all $L$.
\end{proof}


\begin{myclaim}
\label{myclaim3}
Either $F$ is safely manipulable or $F$ returns $X$ whenever voter 1 ranks $X$ first.
\end{myclaim}

\begin{proof}
If there is a profile $Q$ such that $Q_1 = X \cdots$ but $F(Q) \ne X$ then (we show) an escape can be found.  Let $S$ be the profile at which $S_1 = Q_1 = X\cdots$, and for all $i,j > 1$ we have $S_i = S_j = \cdots X$.  If $F$ isn't already safely manipulable then by Claim~\ref{myclaim2} it must be that $F(S) = X$. Therefore, at $S$ every voter other than voter 1 sees their least welcome outcome realised.  Starting at $S$, for $i = 2, 3, 4, \ldots$ sequentially change $S_i$ to $Q_i$ (we do not imply $n \ge 2$ necessarily).  At some stage the value of $F$ will shift away from $X$; an escape is therefore possible under $F$.  
\end{proof}

\begin{corollary}
\label{inversedictatorship}
Either $F$ is safely manipulable or $F(R)=A$ implies voter 1 is ranking $A$ first.
\end{corollary}

\begin{corollary}
\label{WhatOneGets}
Either $F$ is safely manipulable or $F$ always returns either the first or the second choice of voter 1. 
\end{corollary}


In the next two proofs we will repeatedly use Corollary~\ref{WhatOneGets} without explicitly saying so.

\medskip


$F$ is onto, so $F(R) = A$ for some $R$.  If $F$ is not yet safely manipulable then (by Corollary~\ref{inversedictatorship}) $R_1 = AB \cdots C$ (with $B$ second and $C$ last without any loss of generality).  The objects represented by $R$, $A$, $B$, and $C$ shall remain fixed for the remainder of these claims.


\begin{myclaim}
\label{myclaim4}
Either $F$ is safely manipulable or $F_{-C}$ is non-dictatorial.
\end{myclaim}

\begin{proof}
Assume no previous result shows $F$ is safely manipulable.  Voter 1 does not dictate $F$ (no-one does), so there exists a profile $Q$ for which $Q_1 = AX \cdots$ but $F(Q)$ equals $X$ rather than $A$.

Let $S$ be the profile arising when $S_1 = Q_1 = AX \cdots$ and $S_i = S_j = \cdots A$ for all $i,j > 1$.  If $F(S) = A$ then by sequentially changing $S_k,\ k \ge 2$, to $Q_k$ we can find an escape.  So set $F(S) = X$.

Now consider the two-voter profile $( R_1 , S_2 ) =  ( AB\cdots C , \cdots A )$.  If $F_{V_1,V_2}$ maps this profile to $A$ then voter 1 ($1 \in [n]$) can safely manipulate $S$ with $R_1 = AB\cdots C$.  So set $F_{V_1,V_2}( AB\cdots C , \cdots A ) = B$.

Finally, consider the profile $( AB\cdots C , \cdots BAC )$.  If $F_{V_1,V_2}$ maps this profile to $A$, voter $V_2$ can safely manipulate it with $S_2 = \cdots A$; if the profile is mapped to $B$, $F_{-C}$ has no dictator.
\end{proof}


\begin{myclaim}
\label{myclaim5}
Either $F$ is safely manipulable or $F_{-C}$ can reach $A$.
\end{myclaim}

\begin{proof}
Again assume no previous result shows $F$ is safely manipulable.  Construct the profile $S$ by setting $S_1 = R_1 = AB \cdots C$ and for $i,j \ge 2$ setting $S_i = S_j = \cdots ABC$.  If $F(S) = A$ then $F_{-C}$ can reach $A$.  Suppose $F(S) = B$.  For $k \ge 2$, one by one change $S_k$ to $R_k$.  One of these changes must induce the value of $F$ to change from $B$ to $A$; this manipulation will be safe because $F$ will never return $C$ while voter 1's report remains $AB\cdots C$.
\end{proof}


\begin{myclaim}
\label{myclaim6}
We cannot avoid concluding $F$ is safely manipulable.
\end{myclaim}

\begin{proof}
If no earlier result has directly shown $F$ is safely manipulable then $F_{-C}$ is non-dictatorial and can reach $A$, and Claim~\ref{myclaim1} applies.
\end{proof}


We now prove Corollary~\ref{extendGS} from Section~\ref{sec:model}.

\begin{proof}
By Theorem~\ref{ASSWtheorem}, there exists a profile $R$ and a voter $j$ such that, at $R$, $j$ has an incentive to safely strategically vote $L \ne R_j$.  If
\[
F(R_{-\{ j \}}(L)) \succ_{ \{ j \} } F(R)
\]
then we are done.  If not, then $F(R_{-\{ j \}}(L)) = F(R)$.  If $F$ is not anonymous then we cannot guarantee the voter we've identified as $j$ can be the same as the voter we identify as $i$ in the statement of this corollary.  Let $V^*$ be the set of all voters that, at $R$, are of type $R_j$ and have an incentive to (safely or unsafely) strategically vote $L$.  Then let $V_1$ be a maximal element of $V^*$ satisfying both $j \in V_1$ and
\[
F(R_{-V_2} (L)) = F(R)
\]
whenever $j \in V_2 \subseteq V_1$.  Such a set exists because $\{ j \}$ meets the two criteria.  Now $V^* \ne V_1$ for if not, at $R$ the voter $j$ would have no incentive to strategically vote $L$.  So we can find $i \in V^* - V_1$, and for this voter it will be the case that
\[
F(R_{-V_1 \cup \{ i \}}(L)) \succ_{\{i\}} F(R_{-V_1}(L)) = F(R)
\]
and 
\[
F(R_{-V_1 \cup V_3}(L)) \succeq_{\{i\}} F(R_{-V_1}(L)) = F(R)
\]
whenever $i \in V_3 \subseteq V^*-V_1$.  The first line above follows from the properties of $V_1$, and implies $i$ can manipulate alone at $R_{-V_1}(L)$.  The second line follows because $j \in V_1 \cup V_3$ and $L$ is a safe vote for $j$ at $R$.  Let $V^\prime$ comprise of those voters who at $R_{-V_1}(L)$ are of type $R_j$ and have an incentive to strategically vote $L$.  Necessarily, $V^\prime \subseteq V^*$.  More specifically, $V^\prime \subseteq V^* - V_1$, because at $R_{-V_1}(L)$ the voters in $V_1$ are reporting $L$.  The second relation therefore implies $L$ is a safe strategic vote for $i$ at $R_{-V_1}(L)$.
\end{proof}


\section{Discussion}
\label{sec:conclusion}

We summarise our paper, then identify directions future research might take.

We broadened the framework in which Gibbard and Satterthwaite worked so we could look into what can transpire when a group of like-minded voters vote strategically and in unison (while all other voters vote sincerely).  The notion of an incentive to vote strategically was our starting point.  We then distinguished between safe and unsafe strategic votes, the latter being (primarily) those that could potentially lead to strategic overshooting or undershooting.  After providing examples to illustrate and motivate our definitions, we proved our main theoretical result, then used it to extend the Gibbard-Satterthwaite Theorem as follows: we showed that under any onto and non-dictatorial social choice rule, employed to choose one of at least three alternatives, there can arise a situation where a voter can do strictly better by casting a {\em safe} strategic vote than an unstrategic one, provided all other voters are sincere.  

Despite being broader than that of Gibbard and Satterthwaite our framework is still rather narrow. If we extend our framework further a  strategic vote which is safe in our framework may become unsafe. Firstly, voters of other types may not all be sincere. Secondly, and more importantly, voters of the strategic type may have several conflicting safe manipulating moves and again may face inability to coordinate which safe manipulating vote to choose. 
This can be illustrated by the following hypothetical example. Suppose that in a profile there are only two voters 1 and 2 with sincere preference $ABC$. Suppose that they have two safe manipulating moves $L$ and $L'$ and depending on how they vote the results are shown in the following table. 

\begin{center}
\begin{tabular}{cc|ccc}
& & \multicolumn{3}{c}{voter 2} \\
& & $ABC$ & $L$ & $L^\prime$ \\
\hline
\multirow{3}{*}{voter 1} & $ABC$ & $B$ & $B$ & $B$ \\
& $L$ & $B$ & $A$ & $C$ \\
& $L^\prime$ & $B$ & $C$ & $A$
\end{tabular}
\end{center}

We see that if voters 1 and 2 vote sincerely, then the result is $B$, their second best alternative. However if they choose the same manipulating move---either both $L$ or both $L'$---then the outcome is $A$, their most preferred alternative. If they mis-coordinate and choose different manipulating moves, then the outcome is $C$ which is their least desirable alternative. Again for a success of manipulation coordination is necessary.
One particular scenario in which such coordination may occur is when a particular voter (perhaps a public figure)  is capable of sending a single message to the whole electorate (say through the media) calling upon his followers to vote in a certain way.  It is reasonable to suggest that only voters who share his views  might respond and join this loosely formed coalition.  However, not all of them will respond; some of them might consider voting strategically unethical, some just will not get the message. Obviously, to make such a call the public figure must have an incentive to vote strategically and be sure that it is safe to act on it. 
Our study implies that if voters have little ability to communicate with each other then many incentives for voting strategically are relatively weak being unsafe.

\cite{Chamberlin1985} wrote that ``increasingly sophisticated communications technology'' may increase the ``threat'' of manipulations.  Nowadays we live in a much more interconnected world than 25 years ago with many voters being connected by social networks and with a few individuals (mosty celebrities) having millions of followers. 
Have online communication technologies improved the ability of strategic voters to identify and respond appropriately to incentives to vote strategically?  In particular, do such technologies aid like-minded strategic voters attempting to coordinate?  These questions are yet to be studied.

\section*{Acknowledgements}
Arkadii Slinko was supported by the University of Auckland Faculty of Science FRDF grant 3624495/9844.  Shaun White gratefully acknowledges the support of the University of Auckland Faculty of Science.


\bibliographystyle{elsarticle-harv}
\bibliography{Overshooting36}     

\end{document}